\documentclass{article}
\usepackage{times}
\usepackage{graphicx} % more modern
\usepackage{subfigure}

\usepackage{natbib}

\usepackage{algorithm}
\usepackage{algorithmic}

\usepackage{hyperref}

\usepackage[accepted]{icml2017}

\usepackage{amsmath,latexsym,amsbsy,amssymb, amsthm}
\usepackage{bbm}

\newtheorem{Theorem}{Theorem}[section]

\newtheorem{Proposition}{Proposition}[section]

\newtheorem{Assumption}{Assumption}[section]

\newtheorem{Lemma}{Lemma}[section]

\numberwithin{equation}{section}

\newcommand{\h}{\hspace*{.24in}}

\newcommand{\beginsupplement}{%
        \setcounter{table}{0}
        \renewcommand{\thetable}{S\arabic{table}}%
        \setcounter{figure}{0}
        \renewcommand{\thefigure}{S\arabic{figure}}%
     }

\icmltitlerunning{Probabilistic Path Hamiltonian Monte Carlo}

\begin{document}

\twocolumn[
\icmltitle{Probabilistic Path Hamiltonian Monte Carlo}



% you can specify symbols, otherwise they are numbered in order
\icmlsetsymbol{equal}{*}

\begin{icmlauthorlist}
\icmlauthor{Vu Dinh}{equal,fredhutch}
\icmlauthor{Arman Bilge}{equal,fredhutch,UW}
\icmlauthor{Cheng Zhang}{equal,fredhutch}
\icmlauthor{Frederick A. Matsen~IV}{fredhutch}
\end{icmlauthorlist}

\icmlaffiliation{fredhutch}{Program in Computational Biology, Fred Hutchison Cancer Research Center, Seattle, WA, USA}
\icmlaffiliation{UW}{Department of Statistics, University of Washington, Seattle, WA, USA}

\icmlcorrespondingauthor{Frederick, A.~Matsen~IV}{matsen@fredhutch.org}

\icmlkeywords{Hamiltonian Monte Carlo, leapfrog algorithm, phylogenetics}

\vskip 0.3in
]



\printAffiliationsAndNotice{\icmlEqualContribution} % otherwise use the standard text.
%\footnotetext{hi}

% The abstract for ICML needs to be very short according to their guidelines: 1 paragraph, 4-6 sentences.
\begin{abstract}
%Evolutionary tree inference, or phylogenetics, is an essential tool for understanding biological systems from deep-time divergences to recent viral transmission.
%The Bayesian paradigm is now commonly used in phylogenetics to describe support for inferred tree structures or to test hypotheses that can be expressed in phylogenetic terms.
%However, current Bayesian phylogenetic inference algorithms are limited to about 1,000 sequences, which is much fewer than are available via modern sequencing technology.

Hamiltonian Monte Carlo (HMC) is an efficient and effective means of sampling posterior distributions on Euclidean space, which has been extended to manifolds with boundary.
However, some applications require an extension to more general spaces.
For example, phylogenetic (evolutionary) trees are defined in terms of both a discrete graph and associated continuous parameters; although one can represent these aspects using a single connected space, this rather complex space is not suitable for existing HMC algorithms.
In this paper, we develop Probabilistic Path HMC (PPHMC) as a first step to sampling distributions on spaces with intricate combinatorial structure.
We define PPHMC on orthant complexes, show that the resulting Markov chain is ergodic, and provide a promising implementation for the case of phylogenetic trees in open-source software.
We also show that a surrogate function to ease the transition across a boundary on which the log-posterior has discontinuous derivatives can greatly improve efficiency.
\end{abstract}
%This algorithm generalizes previous algorithms by doing classical HMC in the interior of the component Euclidean spaces, but making random choices between alternative paths available at a boundary.

\section{Introduction}

Hamiltonian Monte Carlo is a powerful sampling algorithm which has been shown to outperform many existing MCMC algorithms, especially in problems with high-dimensional and correlated distributions \citep{duane87,neal2011mcmc}.
The algorithm mimics the movement of a body balancing potential and kinetic energy by extending the state space to include auxiliary momentum variables and using Hamiltonian dynamics.
By traversing long iso-probability contours in this extended state space, HMC is able to move long distances in state space in a single update step, and thus has proved to be more effective than standard MCMC methods in a variety of applications.
The method has gained a lot of interest from the scientific community and since then has been extended to tackle the problem of sampling on various geometric structures such as constrained spaces \citep{lan2014spherical, brubaker2012family, hartmann2005constrained}, on general Hilbert space \citep{beskos2011hybrid} and on Riemannian manifolds \citep{girolami2011riemann,wang2013adaptive}.

However, these extensions are not yet sufficient to apply to all sampling problems, such as in \emph{phylogenetics}, the inference of evolutionary trees.
Phylogenetics is the study of the evolutionary history and relationships among individuals or groups of organisms.
In its statistical formulation it is an inference problem on hypotheses of shared history based on observed heritable traits under a model of evolution.
Phylogenetics is an essential tool for understanding biological systems and is an important discipline of computational biology.
The Bayesian paradigm is now commonly used to assess support for inferred tree structures or to test hypotheses that can be expressed in phylogenetic terms \cite{Huelsenbeck2001-zt}.

Although the last several decades have seen an explosion of advanced methods for sampling from Bayesian posterior distributions, including HMC, phylogenetics still uses relatively classical Markov chain Monte Carlo (MCMC) based methods.
This is in part because the number of possible tree topologies (the labeled graphs describing the branching structure of the evolutionary history) explodes combinatorially as the number of species increases.
Also, to represent the phylogenetic relation among a fixed number of species, one needs to specify both the tree topology (a discrete object) and the branch lengths (continuous distances).
This composite structure has thus far limited sampling methods to relatively classical Markov chain Monte Carlo (MCMC) based methods.
One path forward is to use a construction of the set of phylogenetic trees as a single connected space composed of Euclidean spaces glued together in a combinatorial fashion \citep{kim2000slicing, moulton2004peeling, billera2001geometry, Gavryushkin2016-zl} and try to define an HMC-like algorithm thereupon.

\begin{figure}[h!]
\centering
  \includegraphics[width=\linewidth]{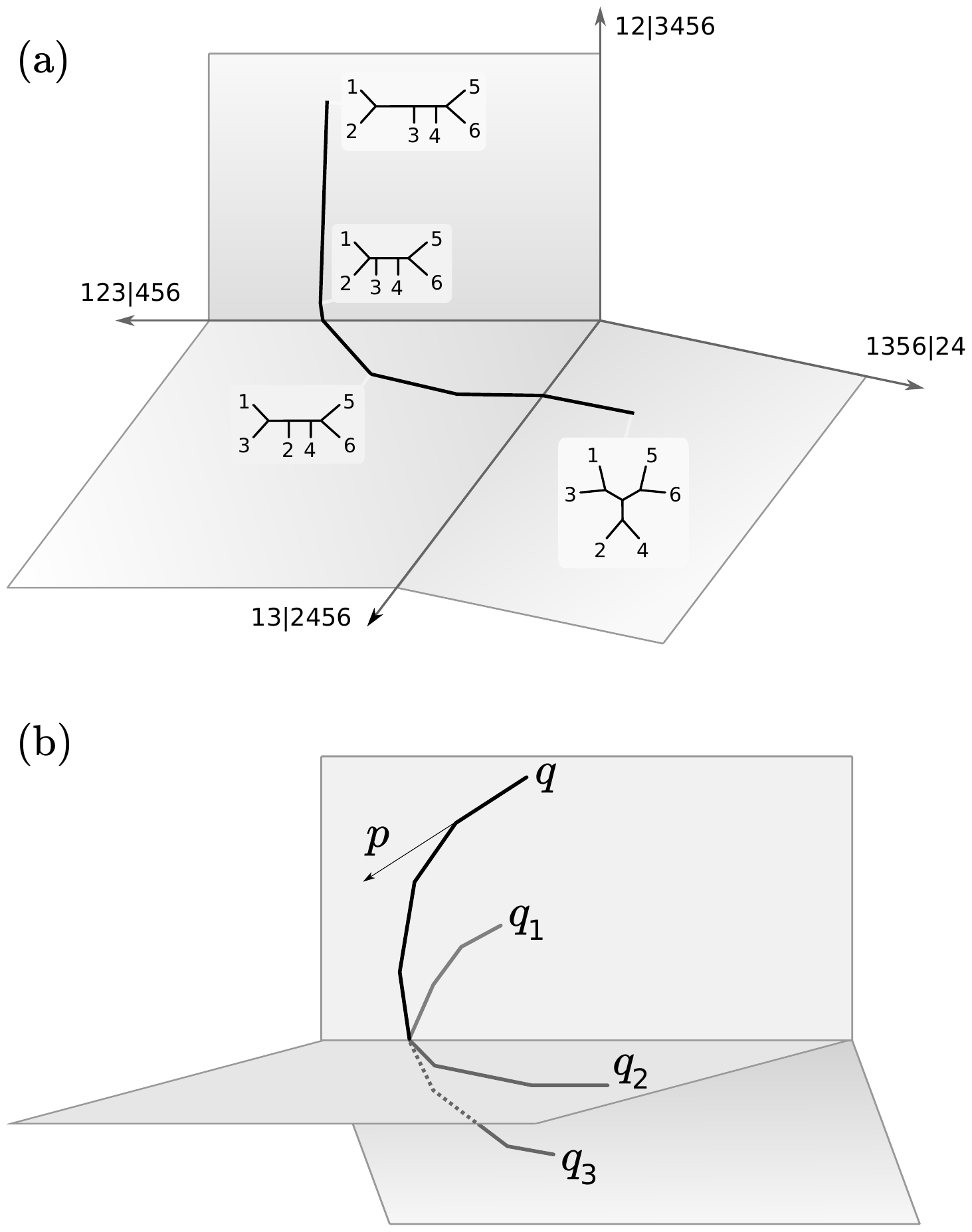}
  \caption{\
PPHMC on the orthant complex of tree space, in which each orthant (i.e.\ $\mathbb{R}_{\ge 0}^n$) represents the continuous branch length parameters of one tree topology.
PPHMC uses the leapfrog algorithm to approximate Hamiltonian dynamics on each orthant, but can move between tree topologies by crossing boundaries between orthants.
(a) A single PPHMC step moving through three topologies; each topology change along the path is one NNI move.
(b) Because three orthants meet at every top-dimensional boundary, the algorithm must make a choice as to which topology to select.
%AB Feels like a confusing use of "neighboring". We are talking about the three neighbors of the degenerate topology, right?
%AB Aha. This uses the definition after assumption 2.1, which is a bit further along.
PPHMC uniformly selects a neighboring tree topology when the algorithm hits such a boundary.
Here we show three potential outcomes $q_1$, $q_2$ and $q_3$ of running a single step of PPHMC started at position $q$ with momentum $p$.
  }
  \label{fig:leapProg}
\end{figure}

Experts in HMC are acutely aware of the need to extend HMC to such spaces with intricate combinatorial structure: \citet{betancourt2017conceptual} describes the extension of HMC to discrete and tree spaces as a major outstanding challenge for the area.
However, there are several challenges to defining a continuous dynamics-based sampling methods on such spaces.
These tree spaces are composed of Euclidean components, one for each discrete tree topology, which are glued together in a way that respects natural similarities between trees.
These similarities dictate that more than two such Euclidean spaces should get glued together along a common lower-dimensional boundary.
The resulting lack of manifold structure poses a problem to the construction of an HMC sampling method on tree space, since up to now, HMC has just been defined on spaces with differential geometry.
Similarly, while the posterior function is smooth within each topology, the function's behavior may be very different between topologies.
In fact, there is no general notion of differentiability of the posterior function on the whole tree space and any scheme to approximate Hamiltonian dynamics needs to take this issue into account.

In this paper, we develop Probabilistic Path Hamiltonian Monte Carlo (PPHMC) as a first step to sampling distributions on spaces with intricate combinatorial structure (Figure~\ref{fig:leapProg}).
After reviewing how the ensemble of phylogenetic trees is naturally described as a geometric structure we identify as an orthant complex \cite{billera2001geometry}, we define PPHMC for sampling posteriors on orthant complexes along with a probabilistic version of the leapfrog algorithm.
This algorithm generalizes previous HMC algorithms by doing classical HMC on the Euclidean components of the orthant complex, but making random choices between alternative paths available at a boundary.
We establish that the integrator retains the good theoretical properties of Hamiltonian dynamics in classical settings, namely probabilistic equivalents of time-reversibility, volume preservation, and accessibility, which combined together result in a proof of ergodicity for the resulting Markov chain.
Although a direct application of the integrator to the phylogenetic posterior does work, we obtain significantly better performance by using a surrogate function near the boundary between topologies to control approximation error.
This approach also addresses a general problem in using Reflective HMC \citep[RHMC;][]{afshar2015reflection} for energy functions with discontinuous derivatives (for which the accuracy of RHMC is of order $\mathcal{O}(\epsilon)$, instead of the standard local error $\mathcal{O}(\epsilon^3)$ of HMC on $\mathbb{R}^n$).
We provide, validate, and benchmark two independent implementations in open-source software.

\section{Mathematical framework}

\subsection{Bayesian learning on phylogenetic tree space}

A \emph{phylogenetic tree} $(\tau, q)$ is a tree graph $\tau$ with $N$ leaves, each of which has a label, and such that each edge $e$ is associated with a non-negative number $q_e$.
Trees will be assumed to be bifurcating (internal nodes of degree 3) unless otherwise specified.
We denote the number of edges of such a tree by $n=2N-3$.
Any edge incident to a leaf is called a \emph{pendant edge}, and any other edge is called an \emph{internal edge}.
Let $\mathcal{T}_N$ be the set of all $N$-leaf phylogenetic trees for which the lengths of pendant edges are bounded from below by some constant $e_0 > 0$.
(This lower bound on branch lengths is a technical condition for theoretical development and can be relaxed in practice.)

We will use nearest neighbor interchange (NNI) moves \citep{robinson1971comparison} to formalize what tree topologies that are ``near'' to each other.
An NNI move is a transformation that collapses an internal edge to zero and then expands the resulting degree 4 vertex into an edge and two degree 3 vertices in a new way (Figure~\ref{fig:leapProg}a).
Two tree topologies $\tau_1$ and $\tau_2$ are called $\emph{adjacent topologies}$ if there exists a single NNI move that transforms $\tau_1$ into $\tau_2$.

We will parameterize $\mathcal{T}_N$ as Billera-Holmes-Vogtmann (BHV) space \citep{billera2001geometry}, which we describe as follows.
An \emph{orthant} of dimension $n$ is simply $\mathbb{R}_{\ge 0}^n$; each $n$-dimensional orthant is bounded by a collection of lower dimensional orthant \emph{faces}.
An \emph{orthant complex} is a geometric object $\mathcal{X}$ obtained by gluing various orthants of the same dimension $n$, indexed by a countable set $\Gamma$, such that: (i) the intersection of any two orthants is a face of both orthants, and (ii) each $x \in \mathcal{X}$ belongs to a finite number of orthants.
Each state of $\mathcal{X}$ is thus represented by a pair $(\tau, q)$, where $\tau \in \Gamma$ and $q \in \mathbb{R}_{\ge 0}^n$.
Generalizing the definitions from phylogenetics, we refer to $\tau$ as its \emph{topology} and to $q$ as the vector of \emph{attributes}.
The topology of a point in an orthant complex indexes discrete structure, while the attributes formalize the continuous components of the space.

For phylogenetics, the complex is constructed by taking one $n$-dimensional orthant for each of the $(2n-3)!!$ possible tree topologies, and gluing them together along their common faces.
The geometry can also be summarized as follows.
In BHV space, each of these orthants parameterizes the set of branch lengths for a single topology (as a technical point, because we are bounding pendant branch lengths below by $e_0$, we can take the corresponding entries in the orthant to parameterize the amount of branch length above $e_0$).
Top-dimensional orthants of the complex sharing a \emph{facet}, i.e.\ a codimension 1 face, correspond to (NNI) adjacent topologies.

For a fixed phylogenetic tree $(\tau,q)$, the phylogenetic likelihood is defined as follows and will be denoted by $L(\tau, p)$ \citep[see][for a full exposition]{Kenney2012-zj}.
Let $\psi = (\psi_1, \psi_2,...,\psi_S) \in \Omega^{N \times S}$ be the observed sequences (with characters in $\Omega$) of length $S$ over $N$ leaves.
The likelihood of observing $\psi$ given the tree has the form
%AB Is it not important that we are missing the stationary probabilities at the "root"? OK if we say likelihood up to a constant factor is:
%V Fixed
\[
L(\tau, q) = \prod_{s=1}^S{\sum_{a^s}{\eta(a_{\rho}^s)\prod_{(u,v)\in E(\tau, q)}{P^{uv}_{a^s_ua^s_v}( q_{uv})}}}
\]
where $\rho$ is any internal node of the tree, $a^s$ ranges over all extensions of $\psi^s$ to the internal nodes of the tree, $a^s_u$ denotes the assigned state of node $u$ by $a^s$, $E(\tau, q)$ denotes the set of tree edges, $P_{ij}(t)$ denotes the transition probability from character $i$ to character $j$ across an edge of length $t$ defined by a given evolutionary model and $\eta$ is the stationary distribution of this evolutionary model.
For this paper we will assume the simplest \citet{JC69} model of a homogeneous continuous-time Markov chain on $\Omega$ with equal transition rates, noting that inferring parameters of complex substitution models is a vibrant yet separate subject of research \citep[e.g.][]{Zhao2016-sk}.

Given a proper prior distribution with density $\pi_0$ imposed on the branch lengths and on tree topologies, the posterior distribution can be computed as $\mathcal{P}(\tau, q) \propto L(\tau, q) \pi_0(\tau, q)$.

\subsection{Bayesian learning on orthant complexes}

With the motivation of phylogenetics in mind, we now describe how the phylogenetic problem sits as a specific case of a more general problem of Bayesian learning on orthant complexes, and distill the criteria needed to enable PPHMC on these spaces.
This generality will also enable applications of Bayesian learning on similar spaces in other settings.
For example in robotics, the state complex can be described by a cubical complexes whose vertices are the states, whose edges correspond to allowable moves, and whose cubes correspond to collections of moves which can be performed simultaneously \citep{ardila2014moving}.
Similarly, in learning on spaces of treelike shapes, the attributes are open curves translated to start at the origin, described by a fixed number of landmark points \citep{feragen2010geometries}.

An orthant complex, being a union of Euclidean orthants, naturally inherits the Lebesgue measure which we will denote hereafter by $\mu$.
%The space itself is a CAT(0) metric space, thus the geodesic connecting any two points in the space is unique.
Orthant complexes are typically not manifolds, thus to ensure consistency in movements across orthants, we assume that local coordinates of the orthants are defined in such a way that there is a natural one-to-one correspondence between the sets of attributes of any two orthants sharing a common face.

\begin{Assumption}[Consistency of local coordinates]
\label{as-coordinate}
Given two topologies $\tau, \tau' \in \Gamma$ and state $x = (\tau, q_{\tau}) = (\tau', q_{\tau'})$ on the boundary of the orthants for $\tau$ and $\tau'$, then $q_{\tau}= q_{\tau'}$.
\end{Assumption}

We show that BHV tree space can be given such coordinates in the Appendix.
For the rest of the paper, we define for each state $(\tau, q) \in \mathcal{X}$ the set $\mathcal{N}(\tau, q)$ of all \emph{neighboring topologies} $\tau'$ such that $\tau'$ orthant contains $(\tau, q)$.
Note that $\mathcal{N}(\tau, q)$ always includes $\tau$, and if all coordinates of $q$ are positive, $\mathcal{N}(\tau, q)$ is exactly $\{\tau\}$.
Moreover, if $\tau' \in \mathcal{N}(\tau, q)$ and $\tau' \ne \tau$, we say that $\tau$ and $\tau'$ are \emph{joined by} $(\tau,q)$.
If the intersection of orthants for two topologies is a facet of each, we say that the two topologies are \emph{adjacent}.

Finally, let $\mathcal{G}$ be the adjacency graph of the orthant complex $\mathcal{X}$, that is, the graph with vertices representing the topologies and edges connecting adjacent topologies.
Recalling that the diameter of a graph is the maximum value of the graph distance between any two vertices, we assume that
\begin{Assumption}
\label{as-diam}
The adjacency graph $\mathcal{G}$ of $\mathcal{X}$ has finite diameter, hereafter denoted by $k$.
\end{Assumption}

For phylogenetics, $k$ is of order $\mathcal{O}(N \log N)$ \cite{Li1996-kc}.

We seek to sample from a posterior distribution $\mathcal{P}(\tau, q)$ on $\mathcal{X}$.
Assume that the negative logarithm of the posterior distribution $U(\tau, q) :=-\log P(\tau, q)$ satisfies:

\begin{Assumption}
\label{as-potential}
$U(\tau, q)$ is a continuous function on $\mathcal{X}$, and is smooth up to the boundary of each orthant $\tau \in \Gamma$.
\end{Assumption}

In the Appendix, we prove that if the logarithm of the phylogenetic prior distribution $\pi_0(\tau, q)$ satisfies Assumption $\ref{as-potential}$, then so does the phylogenetic posterior distribution.
It is also worth noting that while $U(\tau, q)$ is smooth within each orthant, the function's behavior may be very different between orthants and we do not assume any notion of differentiability of the posterior function on the whole space.

\subsection{Hamiltonian dynamics on orthant complexes}
The HMC state space includes auxiliary momentum variables in addition to the usual state variables.
In our framework, the augmented state of this system is represented by the position $(\tau, q)$ and the momentum $p$, an $n$-dimensional vector.
We will denote the set of all possible augmented state $(\tau, q, p)$ of the system by $\mathbb{T}$.

The Hamiltonian is defined as in the traditional HMC setting: $H(\tau, q, p) = U(\tau, q) + K(p)$,
where $K(p) = \frac{1}{2}\|p\|^2$.
We will refer to $U(\tau, q)$ and $K(p)$ as the potential energy function and the kinetic energy function of the system at the state $(\tau,q, p)$, respectively.

Our stochastic Hamiltonian-type system of equations is:
\begin{align}
\begin{split}
\frac{dp_i}{dt} =& - \frac{\partial U}{\partial q_i}(\tau, q) \quad \qquad \ \, \text{if} ~~ q_i > 0 \\
p_i \gets - p_i; \ \ & \tau \sim \, Z(\mathcal{N}(\tau, q))  ~~~ \ \ \ \ \text{if} ~~ q_i = 0  \\
\frac{dq_i}{dt} =& \, p_i
\end{split}
\label{hamilton}
\end{align}
where $Z(A)$ denotes the uniform distribution on the set $A$.

If all coordinates of $q$ are positive, the system behaves as in the traditional Hamiltonian setting on $\mathbb{R}^n$.
When some attributes hit zero, however, the new orthant is picked randomly from the orthants of neighboring topologies (including the current one), and the momenta corresponding to non-positive attributes are negated.

Assumption $\ref{as-potential}$ implies that despite the non-differentiable changes in the governing equation across orthants, the Hamiltonian of the system along any path is constant:
\begin{Lemma}
 $H$ is conserved along any system dynamics.
 \label{conserve}
\end{Lemma}

\subsection{A probabilistic ``leap-prog'' algorithm}
In practice, we approximate Hamiltonian dynamics by the following integrator with step size $\epsilon$, which we call ``leap-prog'' as a probabilistic analog of the usual leapfrog algorithm.
This extends previous work of \cite{afshar2015reflection} on RHMC where particles can reflect against planar surfaces of $\mathbb{R}^n_{\ge 0}$.

In the RHMC formulation, one breaks the step size $\epsilon$ into smaller sub-steps, each of which correspond to an event when some of the coordinates cross zero.
We adapt this idea to HMC on orthant complexes as follows.
Every time such an event happens, we reevaluate the values of the position and the momentum vectors, update the topology (uniformly at random from the set of neighboring topologies), reverse the momentum of crossing coordinates and continue the process until a total step size $\epsilon$ is achieved (Algorithm~\ref{alg}).
We note that several topologies might be visited in one leap-prog step.

If there are no topological changes in the trajectory to time $\epsilon$, this procedure is equivalent to classical HMC.
Moreover, since the algorithm only re-evaluates the gradient of the energy function at the end of the step when the final position has been fixed, changes in topology on the path have no effect on the changes of position and momentum.
Thus, the projection of the particles (in a single leap-prog step) to the $(q,p)$ space is identical to a leapfrog step of RHMC on $\mathbb{R}^n_{\ge 0}$.

\begin{algorithm}
\caption{Leap-prog algorithm with step size $\epsilon$.}
\label{CHalgorithm}
\begin{algorithmic}
\STATE $p \gets p - \epsilon \nabla U(\tau, q) /2$
\IF{$\text{FirstUpdateEvent}(\tau, q,p, \epsilon) = \emptyset$}
\STATE $q \gets q+\epsilon p$
\ELSE
\STATE $t \gets 0$
\WHILE{$\text{FirstUpdateEvent}(\tau, q,p, \epsilon -t) \ne \emptyset$}
\STATE $(q, e, I) \gets \text{FirstUpdateEvent}(\tau, q,p, \epsilon -t)$
\STATE $t \gets t+ e$
\STATE $\tau \sim Z(\mathcal{N}(\tau, q))$
\STATE $p_I \gets - p_I$
\ENDWHILE
\STATE $q \gets q+(\epsilon-t)p$
\ENDIF
\STATE $p \gets p - \epsilon \nabla U(\tau, q) /2$\\
\end{algorithmic}
\label{alg}
\end{algorithm}

Here $\text{FirstUpdateEvent}(\tau, q,p, t)$ returns $x$, the position of the first event for which the line segment $[q, q+ tp]$ crosses zero;  $e$, the time when this event happens; and $I$, the indices of the coordinates crossing zero during this event.
If $q_i$ and $p_i$ are both zero before FirstUpdateEvent is called, $i$ is not considered as a crossing coordinate.
If no such an event exists, $\emptyset$ is returned.

\section{Hamiltonian Monte Carlo on orthant complexes}

Probabilistic Path Hamiltonian Monte Carlo (PPHMC) with leap-prog dynamics iterates three steps, similar to that of classical HMC.
First, new values for the momentum variables are randomly drawn from their Gaussian distribution, independently of the current values of the position variables.
Second, starting with the current augmented state, $s=(\tau, q, p)$, the Hamiltonian dynamics is run for a fixed number of steps $T$ using the leap-prog algorithm with fixed step size $\epsilon$.
The momentum at the end of this trajectory is then negated, giving a proposed augmented state $s^* = (\tau^*, q^*, p^*)$.
Finally, this proposed augmented state is accepted as the next augmented state of the Markov chain with probability $r(s,s^*)= \min \left(1, \exp(H(s) - H(s^*))\right)$.

PPHMC has two natural advantages over MCMC methods for phylogenetic inference: jumps between topologies are guided by the potential surface, and many jumps can be combined into a single proposal with high acceptance probability.
Indeed, rather than completely random jumps as used for MCMC, the topological modifications of HMC are guided by the gradient of the potential.
This is important because there are an enormous number of phylogenetic trees, namely $(2n-3)!!$ trees with $n$ leaves.
Secondly, HMC can combine a great number of tree modifications into a single step, allowing for large jumps in tree space with high acceptance probability.
These two characteristics are analogs of why HMC is superior to conventional MCMC in continuous settings and what we aimed to extend to our problem.

\subsection{Theoretical properties of the leap-prog integrator}
To support the use of this leap-prog integrator for MCMC sampling, we establish that integrator retains analogs of the good theoretical properties of Hamiltonian dynamics in classical settings, namely, time-reversibility, volume preservation and accessibility (proofs in Appendix).

We formulate probabilistic reversibility as:
\begin{Lemma}[Reversibility]
For a fixed finite time horizon $T$, we denote by $P(s, s')$ the probability that the integrator moves $s$ to $s'$ in a single update step.
We have
\[
P((\tau, q, p), (\tau', q', p') ) =P((\tau', q', -p'), (\tau, q, -p) ).
\]
for any augmented states $(\tau', q',p')$ and $(\tau, q, p) \in \mathbb{T}$.
\label{reverse}
\end{Lemma}

The central part of proving the detailed balance condition of PPHMC is to verify that Hamiltonian dynamics preserves volume.
Unlike the traditional HMC setting where the proposal distribution is a single point mass, in our probabilistic setting, if we start at one augmented state $s$, we may end up at countably many end points due to stochastic HMC integration.
The equation describing volume preservation in this case needs to be generalized to the form of Equation $\eqref{eq:vol}$, where the summations account for the discreteness of the proposal distribution.
\begin{Lemma}[Volume preservation]
For every pair of measurable sets $A,B \subset \mathbb{T}$ and elements $s,s' \in  \mathbb{T}$, we denote by $ P(s,s')$ the probability that the integrator moves $s$ to $s'$ in a single update step and define
\[
B(s) = \{  s' \in B: P(s, s') >0 \}
\]
and
\[
A(s') = \{  s \in A: P(s, s') >0 \}.
\]
Then
\begin{equation}
\int_{A} {\sum_{s' \in B(s)}{P(s,s')} ~ds} = \int_{B} {\sum_{s \in A(s')}{P(s',s)} ~ds'}.
\label{eq:vol}
\end{equation}
\label{vol}
\end{Lemma}
If we restrict to the case of trajectories staying in a single topology, $A(s)$ and $B(s')$ are singletons and we get back the traditional equation of volume preservation.
We also note that the measure $ds$ in Equation $\eqref{eq:vol}$ is the Lebesgue measure: when there is no randomness in the Hamiltonian paths, (3.1) becomes the standard volume preservation condition, where volumes are expressed by the Lebesgue measure.

Typically, accessibility poses no major problem in various settings of HMC since it is usually clear that one can go between any two positions in a single HMC step.
In the case of PPHMC, however, the composition of discrete and continuous structure, along with the possible non-differentiability of the potential energy across orthants, make it challenging to verify this condition.
Here we show instead that the PPHMC algorithm can go between any two states with $k$ steps, where $k$ denotes the diameter of adjacency graph $\mathcal{G}$ the space $\mathcal{X}$ and each PPHMC step consists of $T$ leap-prog steps of size $\epsilon$.

\begin{Lemma}[$k$-accessibility]
For a fixed starting state $(\tau^{(0)}, q^{(0)})$, any state $(\tau', q') \in \mathcal{X}$ can be reached from $(\tau^{(0)}, q^{(0)})$ by running $k$ steps of PPHMC.
\label{eqb}
\end{Lemma}

The proof of this Lemma is based on Assumption $\ref{as-diam}$, which asserts that the adjacency graph $\mathcal{G}$ of $\mathcal{X}$ has finite diameter, and that classical HMC allows the particles to move freely in each orthant by a single HMC step.

To show that Markov chains generated by PPHMC are ergodic, we also need to prove that the integrator can reach any subset of positive measure of the augmented state space with positive probability.
To enable such a result, we show:

\begin{Lemma}
For every sequence of topologies $\omega = \{\tau^{(0)}, \tau^{(1)}, \ldots, \tau^{(n_{\omega})}\}$ and every set with positive measure $B \subset \mathcal{X}$, let $B_{\omega}$ be the set of all $(\tau', q') \in B$ such that $(\tau', q')$ can be reached from $(\tau^{(0)}, q^{(0)})$ in $k$ PPHMC steps and such that the sequence of topologies crossed by the trajectory is $\omega$.
We denote by $I_{B, \omega}$ the set of all sequences of initial momenta for each PPHMC step $\{p^{(0)}, \ldots, p^{(k)}\}$ that make such a path possible.

Then, if $\mu(I_{B, \omega})=0$, then $\mu(B_{\omega})=0$.
\label{lips}
\end{Lemma}

We also need certain sets to be countable.

\begin{Lemma}
\label{lem:finite}
Given $s \in \mathbb{T}$, we denote by $R(s)$ the set of all augmented states $s'$ such that there is a finite-size leap-prog step with path $\gamma$ connecting $s$ and $s'$, and by $K(s)$ the set of all such leap-prog paths $\gamma$ connecting $s$ and $s' \in R(s)$.
Then $R(s)$ and $K(s)$ are countable.
Moreover, the probability $P_{\infty}( s, s')$ of moving from $s$ to $s'$ via paths with infinite number of topological changes is zero.
\end{Lemma}

\subsection{Ergodicity of Probabilistic Path HMC}

In this section, we establish that a Markov chain generated by PPHMC is ergodic with stationary distribution $\pi(\tau, q) \propto \exp(-U(\tau, q))$.
To do so, we need to verify that the Markov chain generated by PPHMC is aperiodic, because we have shown $k$-accessibility of the integrator rather than $1$-accessibility.
Throughout this section, we will use the notation $P((\tau, q, p), \cdot)$ to denote the one-step proposal distribution of PPHMC starting at augmented state $(\tau, q, p)$, and $P((\tau, q), \cdot)$ to denote the one-step proposal distribution of PPHMC starting at position $(\tau, q)$ and with a momentum vector drawn from a Gaussian as described above.

We first note that:
\begin{Lemma}
PPHMC preserves the target distribution $\pi$.
\label{invariant}
\end{Lemma}
Given probabilistic volume preservation \eqref{vol}, the proof is standard and is given in the Appendix.

\begin{Theorem}[Ergodic]
\label{ergodic}
The Markov chain generated by PPHMC is ergodic.
\end{Theorem}

\begin{proof}[Proof of Theorem \ref{ergodic}]
For every sequence of topologies $\omega =  \{\tau^{(0)}, \tau^{(1)}, \ldots, \tau^{(n_{\omega})}\}$ (finite by Lemma~\ref{lem:finite}) and every set with positive measure $B \subset \mathcal{X}$, we define $B_{\omega}$ and $I_{B, \omega}$ as in the previous section.
By Lemma $\ref{eqb}$, we have
\[
B = \bigcup_{\omega}{B_{\omega}}.
\]
Assume that $\mu(I_{B, \omega}) = 0$ for all $\omega$.
From Lemma $\ref{lips}$, we deduce that  $\mu(B_{\omega})=0$ for all $\omega$.
This makes $\mu(B) = 0$, which is a contradiction.
Hence $\mu(I_{B, \omega}) > 0$ for some $\omega$ and $P^{n_{\omega}}((\tau^{(0)},q^{(0)}), B )$ is at least the positive quantity
\begin{align*}
\frac{1}{Z} \int_{p \in I_{B, \omega}}{P^{n_{\omega}}((\tau^{(0)}, q^{(0)}, p),  B)~\exp (-K(p))dp}
\end{align*}
where $Z$ is the normalizing constant.
This holds for all sets with positive measure $B \subset \mathcal{X}$, so PPHMC is irreducible.

Now assume that a Markov chain generated by the leapfrog algorithm is periodic.
The reversibility of Hamiltonian dynamics implies that the period $d$ must be equal to 2.
In other words, there exist two disjoint subsets $X_1$, $X_2$ of $\mathcal{X}$ such that $\pi(X_1)>0$, and
\[
P(x, X_2)=1~~ \forall x \in X_1,  ~~\text{and}  ~~P(x, X_1)=1~~ \forall x \in X_2.
\]
Consider $x \in X_1$ with all positive attributes.
There exists a neighborhood $U_x$ around $x$ such that any $y \in U_x$ is reachable from $x$ by Hamiltonian dynamics.
Since $X_1, X_2$ are disjoint, we deduce that $\mu(U_x \cap X_1)=0$.
Since the neighborhood $U_x$ exists for almost every $x \in X_1$, this implies that $\mu(X_1)=0$, and hence, that $\pi(X_1)=0$, which is a contradiction.
We conclude that any Markov chain generated by the leapfrog algorithm is aperiodic.

Lemma $\ref{invariant}$ shows that PPHMC preserves the target distribution $\pi$.
This, along with $\pi$-irreducibility and aperiodicity, completes the proof \citep{roberts2004general}.
\end{proof}

\subsection{An efficient surrogate smoothing strategy}

One major advantage of HMC methods over traditional approaches is that HMC-proposed states may be distant from the current state but nevertheless have a high probability of acceptance.
This partially relies on the fact that the leapfrog algorithm with smooth energy functions has a local approximation error of order $\mathcal{O}(\epsilon^3)$ (which leads to global error $\mathcal{O}(T\epsilon^3)$, where $T$ is the number of leapfrog steps in a Hamiltonian path).

However, when the potential energy function $U(\tau, q)$ is not differentiable on the whole space this low approximation error can break down.
Indeed, although PPHMC inherits many nice properties from vanilla HMC and RHMC \citep{afshar2015reflection}, this discontinuity of the derivatives of the potential energy across orthants may result in non-negligible loss of accuracy during numerical simulations of the Hamiltonian dynamics.
A careful analysis of the local approximation error of RHMC for potential energy functions with discontinuous first derivatives reveals that it only has an local error rate of order at least $\Omega(\epsilon)$ (see proof in Appendix):

\begin{Proposition}
Given a potential function $V$, we denote by $V^+$ and $V^-$ the restrictions of $V$ on the half-spaces $\{x_1 \ge 0\}$ and $\{x_1 \le 0\}$ and assume that $V^+$ and $V^-$ are smooth up to the boundary of their domains.
If the first derivative with respect to the first component of the potential energy $V(q)$ are discontinuous across the hyper-plane $\{x_1=0\}$ (i.e., $(\partial V^+)/(\partial q_1)$ and $(\partial V^-)/(\partial q_1)$ are not identical on this set), then RHMC on this hyper-plane has a local error of order at least $\Omega(\epsilon)$.
\label{exact}
\end{Proposition}

Since PPHMC uses RHMC, when the first derivatives of the potential energy are discontinuous, it also has a global error of order $\mathcal{O}(C \epsilon + T\epsilon^3)$, which depends on the number of \emph{critical events} $C$ along a Hamiltonian path (that is, the number of reflection/refraction events).
This makes it difficult to tune the step size $\epsilon$ for optimal acceptance rate, and requires small $\epsilon$, limiting topology exploration.

To alleviate this issue, we propose the use of a surrogate induced Hamiltonian dynamics \cite{strathmann2015gradient,zhang15} with the Hamiltonian $\tilde{H}(\tau,q,p) = \tilde{U}(\tau,q) + K(p)$, where the surrogate potential energy is
\[
\tilde{U}(\tau,q)  = U(\tau,G(q)),\quad G(q) = (g(q_1),\ldots,g(q_n))
\]
and $g(x)$ is some positive and smooth approximation of $|x|$ with vanishing gradient at $x=0$.
One simple example which will be used for the rest of this paper is
\[
 g_{\delta}(x) = \left\{ \begin{array}{ll}
 x, \; & x\geq\delta\\
 \frac{1}{2\delta}(x^2+\delta^2),\; & 0\leq x <\delta\end{array}
 \right.
\]
where $\delta$ will be called the \emph{smoothing threshold}.

Due to the vanishing gradient of $g$, $\tilde{U}$ now has continuous derivatives across orthants.
However, $\tilde{U}$ is no longer continuous across orthants since $g(0)\neq 0$ and we thus employ the refraction technique introduced in \citet{afshar2015reflection} (see Algorithm \ref{CHalgorithm_surr} for more details).
The proposed state $s_\delta^{\ast} = (\tau_\delta^\ast,q_\delta^\ast,p_\delta^\ast)$ at the end of the trajectory is accepted with probability according to the original Hamiltonian, that is, $\min(1,\exp(H(s)-H(s_\delta^{\ast})))$.

By following the same framework proposed in previous sections, we can prove that the resulting sampler  still samples from the exact posterior distribution $\mathcal{P}(\tau,q)$.
A complete treatment, however, requires more technical adjustments and is beyond the scope of the paper.
%TRIMMABLE
We will leave this as a subject of future work.

\begin{algorithm}
\caption{Refractive Leap-prog with surrogate}
\label{CHalgorithm_surr}
\begin{algorithmic}
\STATE $p \gets p - \epsilon \nabla \tilde{U}(\tau, q) /2$
\IF{$\text{FirstUpdateEvent}(\tau, q,p, \epsilon) = \emptyset$}
\STATE $q \gets q+\epsilon p$
\ELSE
\STATE $t \gets 0$
\WHILE{$\text{FirstUpdateEvent}(q,p, \epsilon -t) \ne \emptyset$}
\STATE $(q, e, I) \gets \text{FirstUpdateEvent}(\tau, q,p, \epsilon -t)$
\STATE $t \gets t+ e$
\STATE $\tau' \sim Z(\mathcal{N}(\tau, q))$
\STATE $\Delta E \gets \tilde{U}(\tau',q) - \tilde{U}(\tau,q)$
\IF{$\|p_I\|^2 > 2\Delta E$}
\STATE $p_I \gets \sqrt{\|p_I\|^2 - 2\Delta E}\cdot \dfrac{- p_I}{\|p_I\|}$
\STATE $\tau \gets \tau'$
\ELSE
\STATE $p_I \gets - p_I$
\ENDIF
\ENDWHILE
\STATE $q \gets q+(\epsilon-t)p$
\ENDIF
\STATE $p \gets p - \epsilon \nabla \tilde{U}(\tau, q) /2$\\
\end{algorithmic}
\label{algWithSurrogate}
\end{algorithm}

As we will illustrate later, compared to the exact potential energy, the continuity of the derivative of the surrogate potential across orthants dramatically reduces the discretization error and allows for high acceptance probability with relatively large step size.

%Compared to the exact potential energy, the continuity of the derivative of the surrogate potential across orthants dramatically reduces the discretization error and allows for high acceptance probability with relatively large step size.
%
%\begin{Lemma}
%Let $A_{\delta} := \{q\in\mathbb{R}^n_{\geq 0}: \exists i \text{  s.t.  } q_i<\delta \}$.
%For any trajectory obtained by simulating the surrogate induced Hamiltonian dynamics, the difference between $H(\tau_0,q_0)$, the total energy at the starting point, and $H(\tau_\delta^{\ast},q_\delta^\ast)$, the total energy at the end, has order
%\begin{align*}
%\Delta H &= H(\tau_\delta^{\ast},q_\delta^\ast) - H(\tau_0,q_0) \\
%&= \mathcal{O}\left(\delta \mathbbm{1}_{A_\delta}(q_\delta^\ast) + \frac{T\epsilon^2}\delta\right).
%\end{align*}
%where $\epsilon$ is the step size and $T$ is the number of leapfrog steps.
%\label{surrogate}
%\end{Lemma}
%Note that the above $\Delta H$ does not depend on the number of critical events, which implies that the acceptance probability could remain high and stable even with lots of reflection and refraction as the sampler explores the tree space.

\section{Experiments}\label{sec:exp}
In this section, we demonstrate the validity and efficiency of our PPHMC method by an application to Bayesian phylogenetic inference.
We compared our PPHMC implementations to industry-standard MrBayes 3.2.5, which uses MCMC to sample phylogenetic trees \cite{Ronquist2012-hi}.
We concentrated on the most challenging part: sampling jointly for the branch lengths and tree topologies, and assumed other parameters (e.g., substitution model, hyper-parameters for the priors) are fixed.
More specifically, for all of our experiments we continued to assume the Jukes-Cantor model of DNA substitution and placed a uniform prior on the tree topology $\tau \sim Z\left(\mathcal{T}_N\right)$ with branch lengths i.i.d.\ $q_i \sim \text{Exponential}\left(\lambda = 10\right)$, as done by others when measuring the performance of MCMC algorithms for Bayesian phylogenetics \citep[e.g.,][]{whidden2015}.
As mentioned earlier, although in the theoretical development we assumed that the lengths of the pendant edges are bounded from below by a positive constant $e_0$ to ensure that the likelihood stays positive on the whole tree space, this condition is not necessary in practice since the Hamiltonian dynamics guide the particles away from regions with zero likelihood (i.e., the region with $U = \infty$).
We validate the algorithm through two independent implementations in open-source software:
\begin{enumerate}
\item a Scala version available at {\small \url{https://github.com/armanbilge/phyloHMC}} that uses the Phylogenetic Likelihood Library\footnote{\url{https://github.com/xflouris/libpll}} \citep{libpll}, and
\item a Python version available at {\small \url{https://github.com/zcrabbit/PhyloInfer}} that uses the ETE toolkit \citep{ete3} and Biopython \citep{bio}.
\end{enumerate}

\subsection{Simulated data}
As a proof of concept, we first tested our PPHMC method on a simulated data set.
We used a random unrooted tree with $N=50$ leaves sampled from the aforementioned prior.
$1000$ nucleotide observations for each leaf were then generated by simulating the continuous-time Markov model along the tree.
This moderate data set provided enough information for model inference while allowing for a relatively rich posterior distribution to sample from.

We ran MrBayes for $10^7$ iterations and sampled every $1000$ iterations after a burn-in period of the first $25\%$ iterations to establish a ground truth for the posterior distribution.
For PPHMC, we set the step size $\epsilon=0.0015$ and smoothing threshold $\delta=0.003$ to give an overall acceptance rate of about $\alpha = 0.68$ and set the number of leap-prog steps $T=200$.
We then ran PPHMC for $10,000$ iterations with a burn-in of $25\%$.
We saw that PPHMC indeed samples from the correct posterior distribution (see Figure~\ref{fig:posterior} in the Appendix).

%\begin{figure}[t!]
%\centering
%
%\includegraphics[width=\linewidth]{Posterior_simuData_JC_long.pdf}
%\caption{Estimated posterior distributions for the top $10$ trees from the ground truth for MrBayes (blue) and PPHMC (orange), respectively.}\label{fig:posterior}
%\end{figure}

\subsection{Empirical data}

We also analyzed an empirical data set labeled DS4 by \citet{whidden2015} that has become a standard benchmark for MCMC algorithms for Bayesian phylogenetics since \citet{Lakner2008-wu}.
DS4 consists of $1137$ nucleotide observations per leaf from $N=41$ leaves representing different species of fungi.
Notably, only $554$ of these observations are complete; the remaining $583$ are missing a character for one or more leaves so the likelihood is marginalized over all possible characters.
\citet{whidden2015} observed that the posterior distribution for DS4 features high-probability trees separated by paths through low-probability trees and thus denoted it a ``peaky'' data set that was difficult to sample from using MrBayes.

To find the optimal choice of tuning parameters for DS4, we did a grid search on the space of step size $\epsilon$ and the smoothing threshold--step size ratio $\delta / \epsilon$.
The number of leap-prog steps $T$ was adjusted to keep the total simulation time $\epsilon T$ fixed.
For each choice of parameters, we estimated the expected acceptance rate $\alpha$ by averaging over 20 proposals from the PPHMC transition kernel per state for 100 states sampled from the posterior in a previous, well-mixed run.
This strategy enabled us to obtain an accurate estimate of the average acceptance rate without needing to account for the different mixing rates in a full PPHMC run due to the various settings for the tuning parameters.

The results suggest that choosing $\delta \approx 2\epsilon$ maximizes the acceptance probability (Figure~\ref{fig:eps_vs_alpha}b).
Furthermore, when aiming for the optimal acceptance rate of $\alpha = 0.65$ \cite{neal2011mcmc}, the use of the surrogate function enables a choice of step size $\epsilon$ nearly 10 times greater than otherwise.
In practice, this means that an equivalent proposal requires less leap-prog steps and gives a more efficient sampling algorithm.

To see the difference this makes in practice, we ran long trajectories for exact and surrogate-smoothed PPHMC with a relatively large step size $\epsilon=0.0008$.
Indeed, we found that the surrogate enables very long trajectories and large number of topology transformations (Figure~\ref{fig:critical}a,c).

\begin{figure}[h!]
\centering
\includegraphics[width=\linewidth]{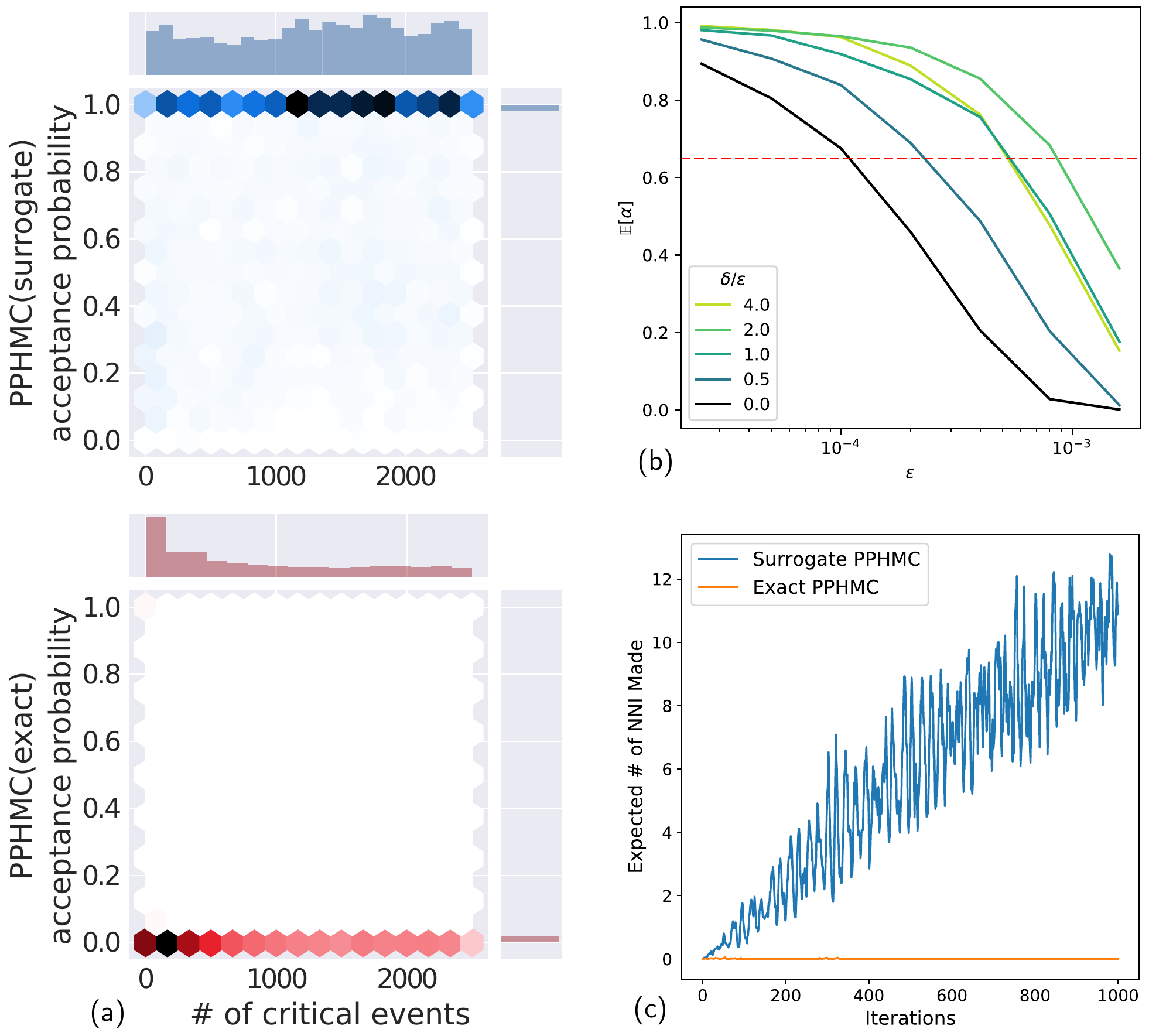}
  \caption{\
(a)
A hexbin plot comparison showing that the surrogate-smoothed PPHMC achieves much higher acceptance rate and longer paths than exact PPHMC on the DS4 data set.
(b)
Average acceptance rate $\alpha$ for different choices of step size $\epsilon$ and smoothing threshold $\delta$ for phylogenetic HMC applied to DS4, where
%the number of leap-prog steps $T$ was chosen to fix the product
$\epsilon T = 0.08$.
The black series ($\delta / \epsilon = 0.0$) is equivalent to
%disabling the surrogate; i.e.,
exact PPHMC.
The dashed red line indicates the
%theoretically-optimal
acceptance rate $\alpha = 0.65$.
(c)
Expected number of NNI moves made by both exact and surrogate-smoothed PPHMC. Results are obtained by averaging 100 HMC iterations with long trajectories ($T=1000$), starting from the same posterior sample.
}
\label{fig:eps_vs_alpha}
\label{fig:critical}
\end{figure}

%\begin{figure}
%\centering
%\includegraphics[width=\linewidth]{eps-vs-alpha.pdf}
%\caption{\
%Average acceptance rate $\alpha$ for different choices of step size $\epsilon$ and smoothing threshold $\delta$ for phylogenetic HMC applied to DS4, where the number of leap-prog steps $T$ was chosen to fix the product $\epsilon T = 0.08$.
%The black series ($\delta / \epsilon = 0.0$) is equivalent to disabling the surrogate; i.e., exact PPHMC.
%The dashed red line indicates the theoretically-optimal acceptance rate $\alpha = 0.65$.
%}
%\label{fig:eps_vs_alpha}
%\end{figure}
%
%\begin{figure}[t!]
%\centering
%\includegraphics[width=.67\linewidth]{DS4_JC_critical.pdf}
%\\
%%  \hspace{.13in}\small{PPHMC(exact)}  \hspace{.67in} \small{PPHMC(surrogate)}
%\caption{A hexbin plot comparison showing that the surrogate-smoothed PPHMC achieves much higher acceptance rate and longer paths than exact PPHMC on the DS4 data set.
%}
%\label{fig:critical}
%\end{figure}

\section{Conclusion}
Sophisticated techniques for sampling posteriors using HMC have thus far been restricted to manifolds with boundary.
To address this limitation, we have developed ``PPHMC,'' which is the first extension of HMC to a space with intricate combinatorial structure.
The corresponding integrator makes a random choice among alternatives when encountering a boundary.
To prove ergodicity, we extend familiar elements of HMC proofs to this probabilistic path setting.
We develop a smoothing surrogate function that enables long HMC paths with many boundary transitions across which the posterior is not differentiable.
Our surrogate method enables high acceptance probability for RHMC \citep{afshar2015reflection} in the case of potential functions with discontinuous derivatives; this aspect of our work is independent of the probabilistic nature of PPHMC.
Our implementation shows good performance on both simulated and real data.
There are many opportunities for future development, including extending the theory to other classes of combinatorially-described spaces and surrogate functions, developing adaptive path length algorithms, as well as extending our implementation for phylogenetics to sample other mutation model and demographic parameters along with topologies.

%\subsection*{Acknowledgements}
%This work supported by NSF grant 1564137.
%The research of Frederick Matsen was supported in part by a Faculty Scholar grant from the Howard Hughes Medical Institute and the Simons Foundation.
%The authors are grateful to Alex Gavryushkin for helpful discussions about this work.

\newpage

\bibliographystyle{plainnat}
{\small
\bibliography{biblio}
}

\clearpage

\beginsupplement

\section{Appendix}

\subsection{Properties of the phylogenetic posterior distribution}

%AB Assumption is about posterior, but we only discuss the likelihood function...
%V Added discussion of the likelihood
\begin{proof}[Assumption $\ref{as-potential}$ for the phylogenetic posterior distribution]

Recall that $L(\tau, q)$ denotes the likelihood function of the tree $T=(\tau, q)$, we have
\[
U(\tau, q) = - \log L(\tau, q) - \log \pi_0(\tau, q)
\]
Since $-\log \pi_0(\tau, q)$ is assumed to satisfy the Assumption $\ref{as-potential}$, we just need to prove that the phylogenetic likelihood function is smooth while each orthant and is continuous on the whole space.

Without loss of generality, we consider the case when a single branch length of some edge $e$ is contracted to zero.
To investigate the changes in the likelihood function and its derivatives, we first fix all other branches, partition the set of all extensions of $\psi$ according to their labels at the end points of $e$, and split $E(T)$ into two sets of edges $E_{\operatorname{left}}$ and $E_{\operatorname{right}}$ corresponding to the location of the edges with respect to $e$.
The likelihood function of the tree $T=(\tau, q)$ can be rewritten as
\begin{eqnarray*}
L(T) &=& \prod_{s=1}^S{\sum_{ij}{\sum_{a \in \mathcal{A}_{ij}}{ \left( \prod_{(u,v) \in E_{\operatorname{left}} }{P^{uv}_{a_ua_v}( q_{uv})}\right) }}}\\
&& \h \ \ \  \times ~ \eta(i) P^{e}_{ij}(t) \times \left( \prod_{(u,v) \in E_{\operatorname{right}} }{P^{uv}_{a_ua_v}( q_{uv})}\right)
\end{eqnarray*}
where $t$ is the branch length of $e$, $\eta$ is the stationary distribution, $\mathcal{A}_{ij}$ denotes the set of all extensions of $\psi$ for which the labels at the left end point and the right end point of $e$ are $i$ and $j$, respectively.
By grouping the products over $E_{\operatorname{left}}$ and $E_{\operatorname{right}}$, the stationary frequency $\eta(\cdot)$, and the sum over $a$ in a single term $b_{ij}^s$, we can define the one-dimensional log-likelihood function as a univariate function of the branch length of $e$
\[
L_T(t)=\prod_{s=1}^S{\left(\sum_{ij}{b^s_{ij}P^{e}_{ij}(t)}\right)}.
\]
Consider the tree $T'$ obtained by collapsing edge $e$ of the tree $T$ to zero.
The likelihood of $T'$ can be written as
\[
L(T') = \prod_{s=1}^S{\left(\sum_{i=j}{b^s_{ij}P^{e}_{ij}(0)}\right)} =\prod_{s=1}^S{\left(\sum_{i}{b^s_{ii}}\right)}
\]
since $P_{ij}(0) = 1$ if $i=j$ and $0$ otherwise.
Thus
\[
\lim_{t\to 0}{L_T(t)} = L(T').
\]
Since this is true for all $(\tau, q)$ and $t \in E(\tau, q)$, we deduce that the likelihood function is continuous up to the boundary of each orthant, and thus, is continuous on the whole tree space.
Moreover, using the same arguments, we can prove that likelihood function is smooth up to the boundary of each orthant.

Now fixing all but two branch lengths $t_e, t_f$, the likelihood can be rewritten as
\[
L_T(t_e, t_f)=\prod_{s=1}^S{\left(\sum_{ij}{b^s_{ij}(t_e)P^{f}_{ij}(t_f)}\right)}
\]
and the derivative of the log likelihood is
\[
\frac{1}{L_T(t_e, t_f) }\frac{\partial L_T}{\partial t_f}(t_e, t_f)= \sum_{s=1}^S{\frac{\sum_{ij}{b^s_{ij}(t_e)(P^{f}_{ij})'(t_f)}}{\sum_{ij}{b^s_{ij}(t_e)P^{f}_{ij}(t_f)}}}.
\]
By using the same argument as above, we have that $b_{ij}^s (t_e)$ is continuous in $t_e$ up to zero and so
\[
\lim_{t_e \to 0}\frac{1}{L_T(t_e, t_f) }\frac{\partial L_T}{\partial t_f}(t_e, t_f) = \frac{1}{L(T') }\frac{\partial L}{\partial t_f}(T').
\]
Thus, when a Hamiltonian particle crosses a boundary between orthants, partial derivatives of the energy function with respect to positive branch lengths are continuous.
\end{proof}

\subsection{Theoretical properties of the leap-prog integrator}

\begin{proof}[Proof of Lemma $\ref{lem:finite}$]
Note that for PPHMC,  in a single leap-prog step $\gamma$ of finite size $\epsilon$, the algorithm only re-evaluates the gradient of the energy function at the end of the step when the final position $q'$ has been fixed, and changes in topology on the path have no effect on the changes of position and momentum.
Thus, the projection $\tilde \gamma$ of $\gamma$ to the $(q,p)$ space is just a deterministic reflected Hamiltonian path.
As a result, for any $s^{(1)} =(\tau^{(1)}, q^{(1)}, p^{(1)}), s^{(2)} =(\tau^{(2)}, q^{(2)}, p^{(2)}) \in R(s)$, we have $(q^{(1)}, p^{(1)})= (q^{(2)}, p^{(2)})$.
This, along with the fact that set of topologies is countable, implies that $R(s)$ is countable.

Now denote by $\{t^{(1)} < t^{(2)} < \ldots < t^{(n)} < \ldots \leq \epsilon \}$ the set of time points at which $\tilde \gamma$ hits the boundary.
Since this set is strictly increasing, it is countable.
Moreover, the $\tau$-component of $\gamma$ is only updated with finite choices at $\{t^{(i)}\}$.
This implies that $K(s)$ is countable.

Finally, consider any leap-prog step $\gamma$ that connects $s$ and $s'$ through infinite number of topological changes.
We note that at each $t^{(i)}$, the next topology is chosen among $x^{(i)} \ge 2$ neighboring topologies.
Denote by $P_{\gamma}( s, s')$ the probability of moving from $s$ to $s'$ via path $\gamma$, we have
\[
P_{\gamma}(s, s') \le \prod_{i=1}^{\infty}{\frac{1}{x^{(i)}}}= 0.
\]
Since $K(s)$ is countable, we deduce that $P_{\infty}( s, s')=0$.
\end{proof}

\begin{proof}[Proof of Lemma $\ref{conserve}$]
Consider any possible path $\gamma$ that connects $s$ and $s'$. By definition, one can find a sequence of augmented states $(s=s^{(0)}, s^{(1)}, s^{(2)}, \ldots, s^{(k)} = s' )$ such that $\gamma$ can be decomposed into segments on which the topology does not change.
From standard result about Hamiltonian dynamics, the Hamiltonian is constant on each segment.

For PPHMC, since the potential energy is continuous across the boundary and the magnitude of the momentum does not change when moving from one orthant to another one, we deduce that the Hamiltonian is constant along that path.

Similarly, for PPHMC with surrogates, the algorithm is designed in such a way that any changes in potential energy is balanced by a change in momentum, which conserves the total energy from one segment to another.
We also deduce that the Hamiltonian is constant along the whole path.
\end{proof}

\begin{proof}[Proof of Lemma $\ref{reverse}$]
Define $\sigma(\tau, q, p):=(\tau, q, -p)$.
Consider any possible leap-prog step $\gamma$ that connects $s$ and $s'$; say the sequences of augmented states $(s=s^{(0)}, s^{(1)}, s^{(2)}, \ldots, s^{(k)} = s' )$, topologies $(\tau=\tau^{(0)}, \tau^{(1)}, \tau^{(2)}, \ldots, \tau^{(k)} = \tau' )$ and times $(t=t^{(0)}, t^{(1)}, t^{(2)}, \ldots,t^{(k)} = t' )$ decompose $\gamma$ into segments on which the topology is unchanged.
Denote by $P_{\gamma}( s, s')$ the probability of moving from $s$ to $s'$ via path $\gamma$, we have
\begin{align*}
P_{\gamma}( s, s') &= \prod_{i}{\mathbb{P}(s^{(i+1)} \mid s^{(i)}, t^{(i+1)} -t^{(i)})} \\
&\h  \times ~\prod_{j}{\mathbb{P}(\tau^{(j+1)} \mid\tau^{(j)})},
\label{eqnprob}
\end{align*}
where each sub-step of the algorithm is a leapfrog update ($\phi^{(i)}$) with some momentum reversing ($\sigma^{(i)}$), that is $s^{(i+1)} = \sigma^{(i)}(\phi^{(i)}(s^{(i)}))$  and $\sigma^{(i)}$ is a map that changes the sign of some momentum coordinates.

If we start the dynamics at $\sigma(s^{(i+1)})$, then since the particle is crossing the boundary, the momenta corresponding to the crossing coordinates are immediately negated and the system is instantly moved to the augmented state
\[
\sigma^{(i)} \sigma(s^{(i+1)}) = \sigma \sigma^{(i)}(s^{(i+1)}) = \sigma(\phi^{(i)}(s^{(i)})).
\]
A standard result about reversibility of Reflective Hamiltonian dynamics implies that the system starting at $\sigma(\phi^{(i)}(s^{(i)}))$ will end at $\sigma(s^{(i)})$ after the same period of time $t^{(i+1)}-t^{(i)}$.
We deduce that
\begin{align*}
&\mathbb{P}(s^{(i+1)} \mid s^{(i)}, t^{(i+1)} -t^{(i)}))\\
=~ &\mathbb{P}(\sigma(s^{(i)}) \mid \sigma(s^{(i+1)}), t^{(i+1)} -t^{(i)})).
\end{align*}
On the other hand, at time $t^{(j)}$, $(\tau^{(j)}, q^{(j)})$ and $(\tau^{(j+1)}, q^{(j)})$ are neighboring topologies, hence
\[
\begin{split}
& \mathbb{P}(\tau^{(j+1)} \mid\tau^{(j)}) = \frac{1}{|\mathcal{N}(\tau^{(j)}, q^{(j)})|} =  \\
& \qquad \frac{1}{|\mathcal{N}(\tau^{(j+1)}, q^{(j)})|} = \mathbb{P}(\tau^{(j)} \mid\tau^{(j+1)}).
\end{split}
\]
Therefore
\[
P_{\gamma}(s, s') = P_{\gamma}(\sigma(s'),  \sigma(s))
\]
for any path $\gamma$.
This completes the proof.
\end{proof}

\begin{proof}[Proof of Lemma $\ref{vol}$]

We denote by $C$ the set of pairs $(s, s') \in A \times B$ such that $P(s, s')>0$.
Let us consider any possible leap-prog step $\gamma$ that connects $s \in A$ and $s' \in B$ crossing a finite number of boundaries and the sequences of augmented states $(s=s^{(0)}, s^{(1)}, s^{(2)}, \ldots, s^{(k)} = s' )$, topologies $(\tau=\tau^{(0)}, \tau^{(1)}, \tau^{(2)}, \ldots, \tau^{(k)} = \tau' )$, times $(t=t^{(0)}, t^{(1)}, t^{(2)}, \ldots,t^{(k)} = t' )$ and indices $\alpha = (\alpha^{(0)}, \alpha^{(1)}, \ldots, \alpha^{(k)})$ (each $\alpha^{(i)}$ is a vector of  $\pm 1$ entries characterizing the coordinates crossing zero in each sub-step) that decompose $\gamma$ into segments on which the topology is unchanged.
By grouping the members of $C$ by the value of $\alpha$ and $\omega$, we have:
\[
C= \bigcup_{(\alpha, \omega)}{C_{\alpha, \omega}}.
\]
Because there will typically be many paths between $s$ and $s'$, the $C_{\alpha, \omega}$ need not be disjoint.
However, we can modify the (countable number of) sets by picking one set for each $(s, s')$ and dropping it from the rest, making a collection of disjoint sets $\{ C_j \}$ such that each $C_j$ is a subset of some $C_{\alpha, \omega}$ and
\[
C= \bigcup_{j \in J}{C_j}.
\]

We will write $s \in A_j(s')$ and $s' \in B_j(s)$ if $(s, s') \in C_j$ and denote
\[
A_{j}=\bigcup_{s' \in B}{A_{j}(s')} \h \text{and} \h B_{j}=\bigcup_{s \in A}{B_{j}(s)}.
\]
We note that although the leap-prog algorithm is stochastic, if $(\alpha, \omega)$ has been pre-specified, the whole path depends deterministically on the initial momentum.
Thus, by denoting the projection of $C_{\alpha, \omega}$ to $A$ by $A_{\alpha, \omega}$, we have that the transformation $\phi_{\alpha, \omega}$ that maps $s$ to $s'$ is well-defined on $A_{\alpha, \omega}$.
Since the projection of the particles (in a single leap-prog step) to the $(q,p)$ space is exactly Reflective Hamiltonian Monte Carlo on $\mathbb{R}^n_{\ge 0}$.
Using Lemma 1, Lemma 2 and Theorem 1 in  \cite{afshar2015reflection}, we deduce that the determinant of the Jacobian of $\phi_{\alpha, \omega}$ is 1.

Now consider any $j\in J$ such that $C_j \subset C_{\alpha, \omega}$.
Because $P(s,s') = P(s',s)$ for all $s, s' \in \mathbb{T}$ and the determinant of the Jacobian of $\phi_{\alpha, \omega}$ is 1, we have
\begin{align}
 \int_{B_j} {\sum_{s \in A_j(s')}{P(s',s)} ~ds'} &=  \int_{B_j} {{P(s',\phi_{\alpha, \omega}^{-1}(s'))} ~ds'}\nonumber \\
&= \int_{A_j} {P(\phi_{\alpha, \omega}(s),s) ~ds}\nonumber \\
&= \int_{A_j} {P(s, \phi_{\alpha, \omega}(s)) ~ds}\nonumber \\
&= \int_{A_j} {\sum_{s' \in B_j(s)}{P(s,s')} ~ds}.
\label{Ai}
\end{align}

Denote
\[
A^*=\bigcup_{j}{A_j} \h \text{and} \h B^*=\bigcup_{j}{B_j}.
\]
Summing $\eqref{Ai}$ over all possible values of $j$ gives
\[
\int_{B^*} {\sum_{s \in A(s')}{P(s',s)} ~ds'} = \int_{A^*} {\sum_{s' \in B(s)}{P(s,s')} ~ds}.
\]
Moreover, we note that for $s \notin A^*$, $B(s)=\emptyset$.
Similarly, if  $s' \notin B^*$, $A(s')=\emptyset$.
We deduce that
\[
\int_{B} {\sum_{s \in A(s')}{P(s',s)} ~ds'} = \int_{A} {\sum_{s' \in B(s)}{P(s,s')} ~ds}.
\]
\end{proof}

\begin{proof}[Proof of Lemma $\ref{eqb}$]
By definition of $k$, for any state $(\tau', q') \in B$, we can find a sequence of topologies $(\tau=\tau^{(0)}, \tau^{(1)}, \tau^{(2)}, \ldots, \tau^{(k)} = \tau' )$ for some $l \le k$ such that $\tau^{(i)}$ and $\tau^{(i+1)}$ are adjacent topologies.
From the construction of the state space, let $(\tau^{(i)}, q^{(i)})$ denote a state on the boundary between the orthants for the two topologies $\tau^{(i)}$ and $\tau^{(i+1)}$.
Moreover, since $(\tau^{(i)}, q^{(i)})$ and $(\tau^{(i+1)}, q^{(i+1)})$ lie in the same orthant, we can find momentum values $p^{(i)}$ and $(p^{(i)})'$ such that
\[
P((\tau^{(i)}, q^{(i)}, p^{(i)}) \to (\tau^{(i+1)}, q^{(i+1)}, (p^{(i+1)})')) >0
\]
for all $i$.
That is, we can get from $(\tau^{(i)}, q^{(i)}, p^{(i)}) $ to $ (\tau^{(i+1)}, q^{(i+1)}, (p^{(i+1)})')$ by a sequence of leapfrog steps $\Sigma^{(i)}$ with length $T$.
By joining the $\Sigma^{(i)}$'s, we obtain a path $\Sigma$ of $k$ PPHMC steps that connects $(\tau^{(0)}, q^{(0)})$ and $(\tau',q')$.
\end{proof}

\begin{proof}[Proof of Lemma $\ref{lips}$]

For a path $\Sigma$ of $k$ PPHMC steps connecting $(\tau^{(0)}, q^{(0)})$ and $(\tau',q')$, we define $F_{\Sigma} = \{(\tau^{(0)}, q^{(0)}), (\tau^{(1)}, q^{(1)}), \ldots, (\tau^{(n_{\omega})}, q^{(n_{\omega})})\}$, where $(\tau^{(i)}, q^{(i)})$ denotes the state on $\Sigma$ that joins the topologies $\tau^{(i)}$ and $\tau^{(i+1)}$.
We first note that although our leap-prog algorithm is stochastic, if the sequence of topologies crossed by a path $\Sigma$ has been pre-specified, the whole path depends deterministically on the sequence of momenta $p=(p^{(0)}, \ldots, p^{(m)})$ along $\Sigma$.
Thus, the functions
\[
\phi_{i, \omega}(p):= q^{(i)} \h \forall p \in I_{B, \omega},
\]
are well-defined.

We will prove that $\phi_{n_{\omega}, \omega}$ is Lipschitz by induction on $n_{\omega}$. For the base case $n_{\omega}=0$, the sequence $\omega$ is of length 1, which implies no topological changes along the path.
The leap-prog algorithm reduces to the baseline leapfrog algorithm and from standard results about HMC on Euclidean spaces \citep[see, e.g.,][]{cances2007theoretical}, we deduce that $\phi_{1, \omega}$ is Lipschitz.

Now assume that the results holds true for $n_{\omega} = n$.
Consider a sequence $\omega$ of length $n+1$.
For all $(\tau', q') \in B_{\omega}$, let $\Sigma(\tau', q')$ be a $(k,T)$-path connecting $(\tau^{(0)}, q^{(0)})$ and $(\tau', q')$.
We recall that
\begin{align*}
&F_{\Sigma(\tau', q')} \\
&= \{(\tau^{(0)}, q^{(0)}), (\tau^{(1)}, \phi_{1, \omega}(p)), \ldots, (\tau^{(n_{\omega})}, \phi_{n_{\omega}, \omega}(p))\},
\end{align*}
where $\phi_{n_{\omega}, \omega}(p))=(\tau', q')$, is the set of states that join the topologies on the path $\Sigma(\tau', q')$.

Define $\omega' = \{\tau^{(0)}, \tau^{(1)}, \ldots, \tau^{(n_{\omega}-1)}\}$ and $B'=\phi_{n_{\omega}-1}(I_{B,\omega})$, the induction hypothesis implies that the function $\phi_{n_{ \omega'}, \omega'} = \phi_{n_{\omega}-1, \omega}$ is Lipschitz on $I_{B', \omega'}=I_{B, \omega}$.

On the other hand, since $(\tau^{(n)}, q^{(n)})$ and  $(\tau^{(n+1)}, q^{(n+1)})$ belong to the same topology, the base case implies that $q^{(n+1)}$ is a Lipschitz function in $p$ and $q^{(n)} = \phi_{n_{\omega}-1, \omega}(p)$.
Since compositions of Lipschitz functions are also Lipschitz, we deduce that $\phi_{n_{\omega}, \omega}$ is Lipschitz.

Since Lipschitz functions map zero measure sets to zero measure sets \citep[see, e.g., Section 2.2, Theorem 2 and Section 2.4, Theorem 1 of][]{evans2015measure}, this implies $\mu(B_{\omega})=0$ which completes the proof.
\end{proof}

\subsection{Ergodicity of PPHMC}

\begin{proof}[Proof of Lemma $\ref{invariant}$ ]

We denote
\begin{align*}
\nu(\tau, q,p) &= \frac{1}{Z} \exp(-U(\tau, q)) \exp \left(-K(p)\right) \\
&= \frac{1}{Z} \exp(-H(\tau, q,p))
\end{align*}
and refer to it as the canonical distribution.

It is straightforward to check that for all $s, s' \in \mathbb{T}$, we have $\nu(s)r(s, s') = \nu(s')r(s' , s)$.
Lemma $\ref{vol}$ implies that
\[
P(s, ds') ds = P(s', ds) ds'
\]
in term of measures.
This gives the detailed balance condition
\begin{align*}
&\int_A \int_B {\nu(s)r(s, s') P(s, ds') ds} \\
= &\int_B \int_A {\nu(s')r(s' , s) P(s', ds) ds'}
\end{align*}
for all $A, B \subset \mathbb{T}$.

We deduce that every update step of the second step of PPHMC satisfies detailed balance with respect to $\nu$ and hence, leaves $\nu$ invariant.
On the other hand, since $\nu$ is a function of $|p|$, the negation of the momentum $p$ at the end of the second step also fixes $\nu$.
Similarly, in the fist step, $p$ is drawn from its correct conditional distribution given $q$ and thus leaves $\nu$ invariant.

Since the target distribution $\pi$ is the marginal distribution of $\nu$ on the position variables, PPHMC also leaves $\pi$ invariant.
\end{proof}

\subsection{Approximation error of reflective leapfrog algorithm}

In this section, we investigate the local approximation error of the reflective leapfrog algorithm \cite{afshar2015reflection} without using surrogates.
Recall that $V^+$ and $V^-$ are the restrictions of the potential function $V$ on the sets $\{x_1 \ge 0\}$ and $\{x_1 \le 0\}$,
and we assume that $V^+$ and $V^-$ are smooth up to the boundary of their domains.

Consider a reflective leapfrog step with potential energy function $V$ starting at $(q^{(0)}, p^{(0)})$ (with $q_1^{(0)}>0$), ending at $(q^{(1)}, p^{(1)})$ (with $q_1^{(1)}<0$) and hitting the boundary at $x$ (with $x_1=0$, i.e., a refraction event happens on the hyper-plane of the first component).

\begin{proof}[Proof of Proposition \ref{exact}]
Let $p$ and $p'$ denote the half-step momentum of a leapfrog step before and after the refraction events, respectively.
Recall that in a leapfrog approximation with refraction at $x_1=0$, we have
\[
p_i^{(0)} = p_i + \frac{\epsilon}{2} \frac{\partial V}{\partial q_i}(q^{(0)}), \h p_i^{(1)} = p_i' - \frac{\epsilon}{2} \frac{\partial V}{\partial q_i}(q^{(1)}),
\]
where
\[
p_1' = \sqrt{p_1^2 - 2dV(x)},
\]
$p_i' = p_i$ for $i>1$, and $dV(x) = V^-(x) - V^+(x)$ denotes the change in potential energy across the hyper-plane.

The change in kinetic energy after this leapfrog step is
\begin{align*}
\Delta K = & -d V(x) - \frac{\epsilon}{2} \sum_i{ \left(p_i\frac{\partial V}{\partial q_i}(q^{(0)}) + p_i' \frac{\partial V}{\partial q_i}(q^{(1)}) \right)} \\
& \quad \ +  \frac{\epsilon^2}{8} \sum_i{\left(\left(\frac{\partial V}{\partial q_i}(q^{(1)}) \right )^2- \left(\frac{\partial V}{\partial q_i}(q^{(0)})\right)^2 \right)}.
\end{align*}
We can bound the second-order term by
\begin{align*}
&\left(\frac{\partial V}{\partial q_i}(q^{(1)}) \right )^2- \left(\frac{\partial V}{\partial q_i}(q^{(0)})\right)^2 \\
&\qquad = 2 \int_{0}^{\epsilon}{\frac{\partial V}{\partial q_i}(q^{(0)} + tp) \frac{\partial^2 V}{\partial q_i^2}(q^{(0)} + tp) p_i ~dt}\\
&\qquad = \mathcal{O}(\epsilon) \cdot \sup_{z, W=V^+, V^-}{\frac{\partial W}{\partial q_i}(z) \frac{\partial^2 W}{\partial q_i^2}(z)}.
\end{align*}
On the other hand for the potential energy,
\begin{align*}
\Delta V &= V(q^{(1)}) - V(q^{(0)}) \\
&= V(q^{(1)}) - V^-(x) + dV(x) +  V^+(x) - V(q^{(0)}) \\
& = dV(x) +  \int_{\epsilon_1}^{\epsilon}{\nabla V(q^{(0)} + tp) \cdot p ~dt}\\
& \h +\int_{0}^{\epsilon_1}{\nabla V(q^{(0)} + tp) \cdot p ~dt}
\end{align*}
where $\epsilon_1$ and $\epsilon_2 := \epsilon - \epsilon_1$  denote the integration times before and after refraction.
By the trapezoid rule for integration,
\begin{align*}
\Delta V = dV(x) &+ \sum_{i>1} \frac{\epsilon}{2}\left(p_i\frac{\partial V}{\partial q_i}(q^{(0)}) + p_i'\frac{\partial V}{\partial q_i}(q^{(1)})  \right) \\
&+ \frac{p_1' \epsilon_2}{2}\frac{\partial V}{\partial q_1}(q^{(1)}) +\frac{p_1' \epsilon_2}{2}\frac{\partial V^-}{\partial q_1}(x) \\
&+~ \frac{p_1\epsilon_1}{2}\frac{\partial V}{\partial q_1}(q^{(0)}) +\frac{p_1 \epsilon_1}{2}\frac{\partial V^+}{\partial q_1}(x)\\
&+ \mathcal{O}(\epsilon^3) \cdot \sup_z {\sum_{i, W=V^+, V^-}{\left(\frac{\partial^3 W}{\partial q_i^3}(z) \right) }}.
\end{align*}
We recall that the error of the trapezoid rule on $[a, b]$ with resolution $h$ is a constant multiple of $h^2(b-a)$, which is of order $\epsilon^3$ in our case.
We deduce that
\begin{align*}
\Delta H= \Delta V + \Delta K = & -\frac{p_1' \epsilon_1}{2}\frac{\partial V}{\partial q_1}(q^{(1)})  +\frac{p_1' \epsilon_2}{2}\frac{\partial V^-}{\partial q_1}(x)\\
& -\frac{p_1 \epsilon_2}{2}\frac{\partial V}{\partial q_1}(q^{(0)})  +\frac{p_1 \epsilon_1}{2}\frac{\partial V^+}{\partial q_1}(x)  \\
&+ \mathcal{O}(\epsilon^3).\end{align*}
Using Taylor expansion, we have
\[
\frac{\partial V}{\partial q_1}(q^{(1)}) = \frac{\partial V^-}{\partial q_1}(x) + \mathcal{O}(\epsilon),
\]
and
\[
\frac{\partial V}{\partial q_1}(q^{(0)})= \frac{\partial V^-}{\partial q_1}(x)  +  \mathcal{O}(\epsilon).
\]
This implies
\[
\Delta H =(\epsilon_2 - \epsilon_1) \left(p_1' \frac{\partial V^-}{\partial q_1}(x) - p_1\frac{\partial V^+}{\partial q_1}(x) \right)  + \mathcal{O}(\epsilon^2).
\]
In general, there is no dependency between $\epsilon_1$ and $\epsilon_2$, and the only cases where $\Delta H$ is not of order $\mathcal{O}(\epsilon)$ are when
\[
 \sqrt{p_1^2 - 2dV(x)}\, \frac{\partial V^-}{\partial q_1}(x) - p_1\frac{\partial V^+}{\partial q_1}(x) =0.
\]
In order for this to be true for all $p$, we need to have either
\[
dV(x)=0 \h \text{and} \h \frac{\partial V^-}{\partial q_1}(x) = \frac{\partial V^+}{\partial q_1}(x),
\]
or
\[
\frac{\partial V^-}{\partial q_1}(x) = \frac{\partial V^+}{\partial q_1}(x) = 0.
\]
In both cases, the first derivative of $V$ with respect to the first component must be continuous.
This completes the proof.
\end{proof}

\subsection{Estimated posterior tree distributions for the simulated data}

\begin{figure}[h!]
\centering

\includegraphics[width=\linewidth]{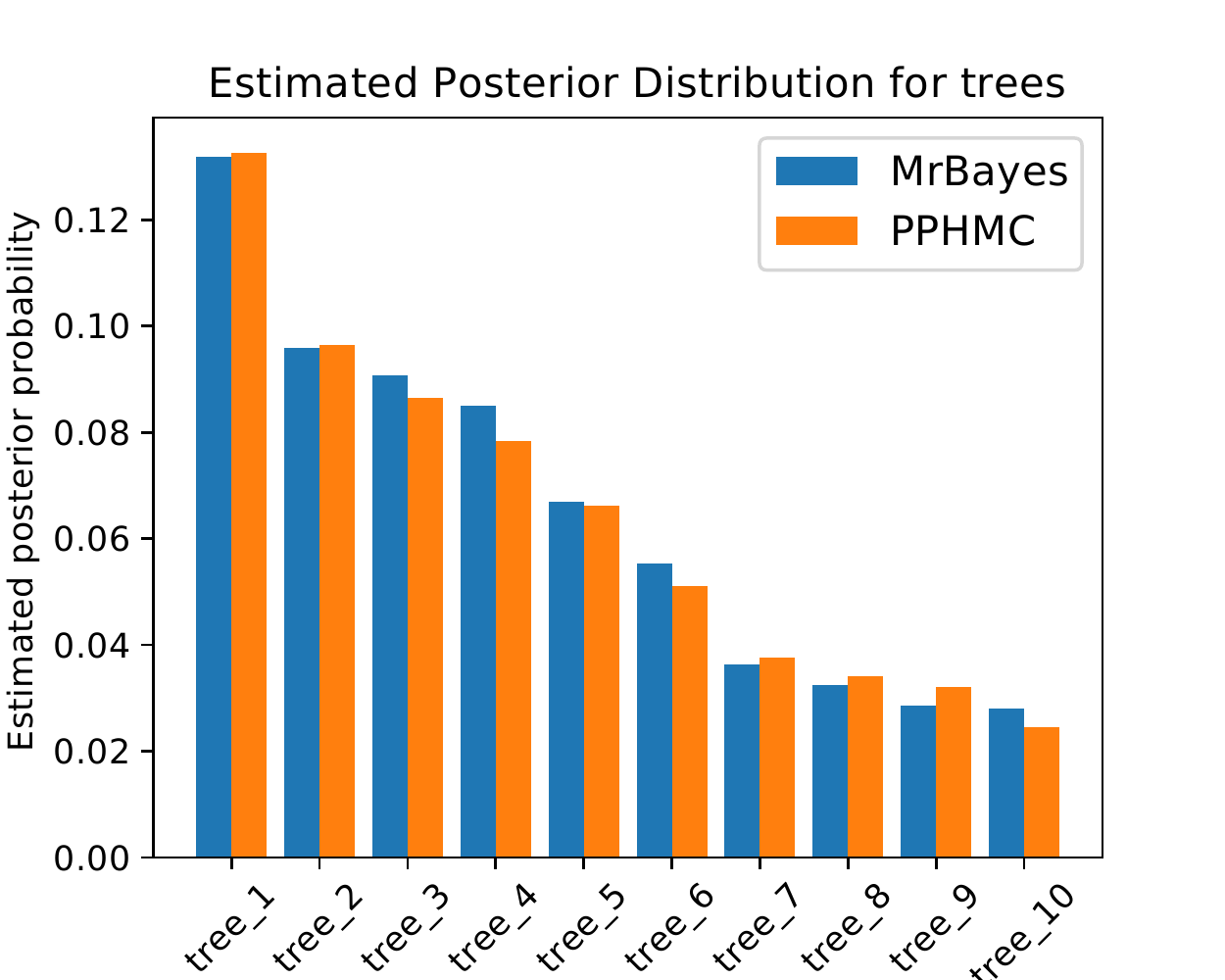}
\caption{Estimated posterior distributions for the top $10$ trees from the ground truth for MrBayes (blue) and PPHMC (orange), respectively.}\label{fig:posterior}
\end{figure}

\subsection{Coordinate systems for branch lengths on trees}

In this section we verify Assumption~\ref{as-coordinate} for phylogenetic trees.
Further explanation of the framework used here can be found in \citep{Semple2003-em,Bryant2004-cw}.

Assume we are considering phylogenetic trees on $N$ leaves, and that those leaves have labels $[N] := \{1, \ldots, N\}$.
Every possible edge in such a phylogenetic tree can be described by its corresponding \emph{split}, which is a partition of $[N]$ into two non-empty sets, by removing that edge of the tree and observing the resulting partitioning of the leaf labels.
If a split can be obtained by deleting such an edge of a given phylogenetic tree, we say that the tree \emph{displays} that split.
We use a vertical bar ($|$) to denote the division between the two sets of the bipartition.
For example, if we take the unrooted tree with four leaves such that 1 and 2 are sister to one another, the tree displays splits $1|234$, $12|34$, $134|2$, $124|3$, and $123|4$.
Two splits $A|B$ and $C|D$ on the same leaf set are called \emph{compatible} if one of $A \cap C$, $B \cap C$, $A \cap D$, or $B \cap D$ is empty.
A set of splits that are pairwise compatible can be displayed on a phylogenetic tree \citep{Buneman1971-mb}, and in fact the set of pairwise compatible sets of splits is in one-to-one correspondence with the set of (potentially multifurcating) unrooted phylogenetic trees.

When a single branch length goes to zero, $\mathcal{N}(\tau, q)$ will have three elements: $\tau$ itself and its two NNI neighbors.
When multiple branch lengths go to zero, one can re-expand branch lengths for any set of splits that are compatible with each other and with the splits that did not originally go to zero.
This generalizes the NNI condition.
However, the correspondence between the branches that went to zero and the newly expanded branches is no longer obvious.

One can define such a correspondence using a global splits-based coordinate system.
Namely, such a coordinate system can be achieved by indexing branch length vectors by splits, with the proviso that for any two incompatible splits $r$ and $s$, one of $q_r$ or $q_s$ is zero.
We could have used such a coordinate system for this paper, such that branch length vectors $q$ would live in $\mathbb{R}^{2^{N-1}}$.

However, for simplicity of notation, we have indexed the branch lengths (e.g.\ in Algorithm~\ref{alg}) with integers $[n]$ corresponding to the actual branches of a phylogenetic tree.
Thus our branch length vectors $q$ live in $2N-3$ dimensions.
One can use a total order on the splits to unambiguously define which branches map to which others when the HMC crosses a boundary.
We will describe how this works when two branch lengths, $q_i$ and $q_j$, go to zero.
The extension to more branch lengths is clear.

Our branch indices $i, j \in [2N-3]$ are always associated with a phylogenetic tree $\tau$ with numbered edges.
For any branch index $i$ on $\tau$, one can unambiguously take the split $s_i$.
Assume without loss of generality that $s_i < s_j$ in the total order on splits.
Now, when $q_i$ and $q_j$ go to zero, one can transition to a new tree $\tau'$ which may differ from $\tau$ by up to two splits.
We assume without loss of generality that these are actually new splits (if not, we are in a previously defined setting) which we call $s_1'$ and $s_2'$ such that $s_1' < s_2'$.
We carry all of the branch indices for branches that aren't shrinking to zero across to $\tau'$.
Then map branch $i$ in $\tau$ to the branch in $\tau'$ corresponding to the split $s_1'$, and branch $j$ to the branch in $\tau'$ corresponding to the split $s_2'$.
Thus, for example, the momentum $q_i$ in the $\tau$ orthant is carried over to this corresponding $q_i$ in the $\tau'$ orthant.
\end{document}